\documentclass[12pt]{article}

\usepackage{a4wide}

\usepackage{amsfonts,amssymb,amstext,amsmath}
\usepackage[latin1]{inputenc}
\usepackage[english]{babel}
\usepackage{url}

\usepackage{newproof}
\usepackage[thmmarks]{ntheorem}
\theorembodyfont{\rmfamily}
\theoremstyle{break}
\theoremsymbol{\ensuremath{\dashv}}

\newtheorem{theorem}{Theorem}
\newtheorem{definition}[theorem]{Definition}
\newtheorem{proposition}[theorem]{Proposition}
\newtheorem{corollary}[theorem]{Corollary}
\newtheorem{lemma}[theorem]{Lemma}

\begin{document}

\title{A secure additive protocol for card players}

\author{Andr\'es Cord\'on--Franco, Hans van Ditmarsch, David Fern\'andez--Duque,\\ Joost J.\ Joosten, and Fernando Soler--Toscano\thanks{Emails and affiliations: {\tt \{acordon,hvd,dfduque,fsoler\}@us.es}, University of Sevilla, Spain; and {\tt jjoosten@ub.edu}, University of Barcelona, Spain. Hans van Ditmarsch is affiliated to IMSC (the Institute of Mathematical Sciences Chennai), India, as a research associate. He is corresponding author. This work extends the results reported in \cite{albertetal:2011}.}}
\date{}

\maketitle

\begin{abstract}
Consider three players Alice, Bob and Cath who hold $a$, $b$ and $c$ cards, respectively, from a deck of $d=a+b+c$ cards. The cards are all different and players only know their own cards. Suppose Alice and Bob wish to communicate their cards to each other without Cath learning whether Alice or Bob holds a specific card.

Considering the cards as consecutive natural numbers $0,1,...$, we investigate general conditions for when Alice or Bob can safely announce the sum of the cards they hold modulo an appropriately chosen integer. We demonstrate that this holds whenever $a,b>2$ and $c=1$. Because Cath holds a single card, this also implies that Alice and Bob will learn the card deal from the other player's announcement. 

\end{abstract}

\section{Introduction}

Consider players Alice and Bob each drawing three cards from a deck of seven cards numbered 0, 1, ..., 6, while Cath, a third player acting as the eavesdropper, gets the remaining card. One way for Alice and Bob to communicate their cards to each other by way of public announcements, without informing Cath of any of their cards, is when both announce the sum modulo 7 of their cards. For example, suppose Alice holds cards 0, 1, and 2, and Bob cards 3, 4, and 5 (so that Cath holds card 6). Alice would then announce ``3'' and Bob would announce ``5''. There are five different hands of cards compatible with Alice's announcement, namely 012, 046, 136, 145, and 235. From those, Cath, who holds 6, can eliminate 046 and 136. Cath cannot deduce that Alice has 0, because 145 is a possible hand of Alice. She also cannot deduce that Alice does not have 0, because Alice's actual hand 012 is also a possible hand. Bob's announcement that the sum of his cards modulo 7 is 5, does not provide additional information. But this is informative for Alice, namely in order to learn Bob's cards, Cath already knows Bob's sum after Alice's announcement, because she can derive that from: the sum of all cards, her own card, and the sum of Alice's cards. This solution consisting of two modulo-sum announcements is an instance of what we will later introduce as the $7$-$\mathsf{ModSum}$ protocol.

The modulo-sum solution for this problem was among the answers to a Moscow Mathematics Olympiad problem \cite{makarychevs:2001}. In that version, after Alice's modulo sum announcement, Bob announces Cath's card. This answer can easily be shown to have the same information content as a consecutive modulo sum announcement by Bob. The problem itself is much older and originates with Kirkman \cite{kirkman:1847}. There, the solution takes the form of a {\em design}, a collection of subsets of a given set that satisfies certain regulaties \cite{stinson:2004}. The five alternative hands of Alice solving the problem above are not a design. But, for the same card deal as above, Alice's announcement could also have been that she holds one of seven triples $012, 034, 056, 135, 146, 236, 245$. We note that each card occurs three times in this announcement and each pair of cards once. It is a design. Instead of 3, 3, and 1 cards, players Alice, Bob and Cath may hold $a,b,c$ cards. Some general patterns and special cases of card deal sizes $(a,b,c)$ for which Alice and Bob can communicate their hands of cards to each other are found in \cite{albertetal:2005}. A complete characterization is not known. The methods employed in that publication are not numerical (they do not involve modulo sum announcements). 

There is also a relation to the {\em bit exchange} problem: is it possible for Alice and Bob to share a secret bit by public communication? A secret bit is the value of a proposition commonly known by them, but not know by the `eavesdropper' Cath. For example, let $i$ be a card held by Alice, then after the execution of the above protocol Alice and Bob share the value of the secret bit `Alice holds card $i$'. The seminal publication for card exchange and bit exchange protocols is \cite{fischeretal:1996}, and additional contributions include \cite{mizukietal:2002}. A (single) bit exchange is easier to accomplish: it suffices to keep the ownership of one card a secret, not necessarily the ownership of all cards. E.g., if Alice were to announce `I have one of 012, 034, 056', then afterwards Alice and Bob share the value of the proposition `Alice holds 1'. But, of course, it is now known to all that Alice holds card 0.

The crucial difference between our approach and the ones cited above is that we begin by assigning an algebraic structure to the set $\{0,\dots,d-1\}$. This allows us to describe protocols in a compact fashion and draw on known combinatorial results, namely, a conjecture of Erd\"os and Heilbronn \cite{eh} giving lower bounds of the number of distinct sums of $n$-subsets of $\mathbb Z/(d)$ (later proven by Dias da Silva and Hamidourne \cite{dh}), and Bertrand's postulate \cite{bertrand}, as well as a theorem of Nagura \cite{nagura:1952}, which show that one can always find a prime number relatively close to a given integer. Below we will formulate these results in full detail.

\bigskip

Referring to the agents as `card players' is suggestive, because the scenarios we model are not meant to model any specific game, and the condition of Alice and Bob's full knowledge together with Cath's total ignorance would not be too appropriate, say, for poker. Rather, the interest of the Russian cards problem is that it gives {\em unconditionally safe} information exchange protocols, as opposed to {\em conditionally safe} cryptography which depends on the assumption that certain problems are of untractably high complexity, such as prime decomposition. In such protocols, Alice and Bob would be the principals involved in the exchange of information, and Cath the spy or eavesdropper (a.k.a.\ Eve) intercepting their communications. The authors believe that their work could lead to new methods of secure information exchange.

\bigskip

In this paper we give general conditions on $a,b,c$ that guarantee that Alice can safely announce the sum of her cards modulo the smallest prime greater than or equal to $d=a+b+c$ (that is, without Cath learning whether Alice or Bob holds any given card). In particular, they hold whenever $a,b>2$ and $c=1$, except for the cases $(3,4,1)$ and $(4,3,1)$. If $c=1$ this announcement is also informative, i.e., Alice and Bob will learn the card deal from the sum announcement of the other player. In \S\ref{sone} we consider the case where the sum $d$ of all cards is already prime, and in \S\ref{two} we consider composite $d$. The exceptional cases $(3,4,1)$ and $(4,3,1)$ are treated separately in \S\ref{341}. This is then followed by the main theorem which incorporates all of these partial results.

\section{Secure and informative protocols} \label{szero}

Consider a deck $D$ of $d$ cards. We will identify the cards with the natural numbers $0,1,\dots,d-1$. We say that a deal of $a$ cards for Alice, $b$ for Bob and $c$ for Cath, with $a,b,c\geq 1$ and $a+b+c=d$  has {\em card deal size} $(a,b,c)$. $A$, $B$ and $C$ represent each player's hand, i.e., the cards that they hold. We can represent the card deal by the collection $(A,B,C)$.

Consider the domain consisting of all card deals of a certain size $(a,b,c)$, i.e., of all $(A,B,C)$ with $A,B,C \subseteq D$, and $A,B,C$ pairwise disjoint. The property that players initially only know their own cards induces a partition on this domain. The card deals in the same equivalence class for a player are {\em indistinguishable} for that player. A {\em protocol} consists of a finite sequence of announcements by Alice and Bob. These announcements are supposed to be truthful. An announcement by, e.g., Alice corresponds to a collection of possible hands of cards $A$ including her actual hand, and the effect of her announcement is the restriction of the domain of card deals to all deals wherein she has one of those hands. Subsequent announcements induce successive such restrictions on the domain of card deals. Given a deal of cards and a sequence of announcements, a player $X$ {\em knows} that a player $Y$ holds a card $z$, if and only if $Y$ holds $z$ in all card deals indistinguishable for player $X$. A player knows the card deal if she knows the ownership of all cards.

E.g., before Alice's announcement 012 034 056 135 246 above, Bob, who holds $\{3,4,5\}$, does not know that Alice holds 0, because he considers hands $\{0,1,2\}$ and $\{1,2,6\}$ (among more) possible for Alice. But after the announcement, Bob knows that Alice holds 0. Before announcements, all players know the cards in their own hand.

In this work we restrict ourselves to protocols consisting of two announcements. An announcement by Alice is {\em informative} if, whenever Alice can make the announcement, Bob knows the card deal afterwards. An announcement by Alice is {\em secure} if, whenever Alice can make the announcement, Cath does not know any of Alice or Bob's cards afterwards. Similarly for informative and secure announcements by Bob.

\section{The modulo-sum announcement} \label{nogeens}

For $n\geq d$ and $X\subseteq D$, $\sum_n X$ denotes the sum of the elements of $X$ seen as elements of $\mathbb Z/(n)$. 

\begin{definition}[$n$-$\mathsf{ModSum}$] Let $n\geq d$ be an integer. We define the $n$-modulo-sum ( $n$-$\mathsf{ModSum}$) protocol to be that where Alice and Bob successively announce the sum of their respective cards modulo $n$.
\end{definition}
As they announce a sum, we call it an additive protocol.
\begin{proposition}[Information] \label{information}
The $n$-$\mathsf{ModSum}$ protocol is informative if $c=1$.
\end{proposition}

\begin{proof}
We have that
\[\sum_nA+\sum_nB+\sum_nC=\sum_nD.\]
Once Alice announces $\sum_nA$, Bob can compute the term $\sum_nC$ from the three other terms in this equation. In case that $C$ is a singleton, this immediately gives him the value of Cath's card. Similarly, $\sum_nB$ allows Alice to figure out Cath's card as well, giving both Alice and Bob knowledge of the entire deal.
\end{proof}

The order in which the announcements are made is unimportant, since both agents announce information that was true at the beginning and that remains true after the announcement by the other agent. Although obvious in this case, it is a remarkably characteristic of a protocol: the order of announcements in a protocol can typically not be altered. E.g., in the protocol where Bob's response to Alice's sum announcement is to announce Cath's card, the order of announcements cannot be swapped.

From Cath's perspective, it does not matter whether she learns $\sum_nA$ or $\sum_nB$, since she can use one to compute the other. In that sense she does not learn anything new in the second announcement. Suppose Alice is first, then Cath already knows what Bob will announce after Alice, before he is doing so. She cannot eliminate card deals given Bob's announcement. But in another sense she still learns something from Bob's announcement: after that, Cath knows that Alice knows the card deal.

If $c>1$, the announcement is typically not informative. For example, if Alice, Bob and Cath each hold two cards respectively, namely $\{0,1\}$, $\{2,5\}$, and $\{3,4\}$, then Alice announcing that $\sum_6 A = 1$ does not result in Bob learning Alice's cards, as they may also be $\{3,4\}$.

Although the $n$-$\mathsf{ModSum}$ protocol is informative for $c=1$, it is not always secure. For a trivial example, if Alice, Bob and Cath each hold one card, then Cath learns the entire deal from Alice's announcement of her card.

The following proposition is a special case of a similar proposition in \cite{albertetal:2005}.

\begin{proposition}[Reduction] \label{reduction} For $n\geq d$, if the announcement of $\sum_n A$ is secure for deals of size $(a,b,c)$, then the announcement of $\sum_n A$ is also secure for deals of size $(a,b,c')$ with $c' \leq c$.
\end{proposition}

\begin{proof} If for some card deal of size $(a,b,c')$, after the announcement of $\sum_n A$, Cath knows some of Alice or Bob's cards, then this would be also the case for the card deal of size ($a,b,c)$ in which Cath additionally holds $c -c'$ (new) cards
\[(a+b+c'), (a+b+c'+1), \dots ,(a+b+c-1)\] and where it is publicly known that she holds these additional cards.
\end{proof}

\section{The number-of-cards modulo-sum announcement} \label{sone}

The case of the $n$-$\mathsf{ModSum}$ protocol where $n = d$, i.e., the total number of cards, we give a special name.

\begin{definition}[$\mathsf{DModSum}$] The number-of-cards modulo-sum ($\mathsf{DModSum}$) protocol is that where Alice and Bob announce the sum of their respective cards modulo $d$.
\end{definition}

With the following lemmas, we will work our way towards the main result of this section that for card deals of size $(a,b,1)$ with $a+b+1$ prime, announcing $\sum_dA$ is secure for $a,b>2$ (Corollary \ref{corC1}). (Given these results, they dually hold for $\sum_dB$ as well.)

\begin{lemma} \label{lem1}
Let $t,d$ be natural numbers with $0<t<d-2$ and $x\in \mathbb Z/(d)$. Then, there exists $T\subseteq \mathbb Z/(d)$ with $|T|=t$ and $\sum T=x$.
\end{lemma}

\begin{proof}
There are at least $\left\lceil\frac{d-4}{2}\right\rceil$ pairs of the form $\{+ c,-c\}$ with $c\not=-c$ and $c\not=\pm x$. The reason for this is that the only elements that do not belong to one such pair are $0$, $\pm x$ and, if $d$ is even, $d/2$.

Let $+ c_1,-c_1, \dots$ be an enumeration of all these pairs. Then, if $t$ is odd we can take
\[T_t=\{x\}\cup\left\{+ c_i,-c_i:i<\frac{t-1}2\right\};\]
note that $\frac{t-1}2\leq\left\lceil\frac{d-4}{2}\right\rceil$ because $t<d-2$.

Otherwise, take $T_t=T_{t-1}\cup\{0\}$.
\end{proof}

\begin{proposition} \label{teoAB} Announcing $\sum_dA$ (or, equivalently, $\sum_dB$) is secure for all deals of size $(a,b,c)$ if and only if, for any value of $0\leq x < d$ and any set $S\subset D$ of $a+b-1$ cards, there are $A,B\subset S$ such that $|A|=a$, $|B|=b$ and $\sum_d A = \sum_d B = x$.
\end{proposition}

\begin{proof}
First let us prove the `if' direction, and assume the right-hand condition holds.

Suppose that Alice announces $\sum_dA$ and $z$ is any card that Cath does not hold. We must show that Cath cannot determine whether Alice or Bob holds $z$.

Consider the set $S=D\setminus(C\cup \{z\})$. This set has $a+b-1$ elements.

Now, pick $A'\subset S$ such that $\sum_dA'=\sum_dA$ and $|A'|=a$.

In the deal where Alice holds $A'$, she does not hold $z$ yet the sum of her cards is still $\sum_dA$, so Cath cannot distinguish between this scenario and the current one. Therefore Cath cannot correctly conclude that Alice holds $z$.

Similarly, we can pick $B'\subset S$ with $|B'|=b$ and and $\sum_dB'=\sum_dB$. Again Cath cannot distinguish between this deal and the original deal (since the sum of Alice's cards does not change either), and Bob does not hold $z$ in this new deal. Hence, Cath cannot correctly conclude that Bob holds $z$.

For the other direction, suppose that there is $0\leq x < d$ and a particular set $S$ of $a+b-1$ cards such that there does not exist $A\subset S$ with $|A|=a$ and $\sum_dA = x$. We must show that announcing $\sum_dA$ is not always secure. (The situation is symmetric if the condition fails for $b$ instead of $a$, so we only treat this case.) 

For such an $S$ there exists a minimal subset $T\subset D\setminus S$ such that there is $A\subset S\cup T$ with $|A|=a$ and $\sum_d A = x$ (Lemma \ref{lem1} guarantees that at least one such set exists, and therefore there is a minimal one). Observe that $|S\cup T|\geq a+b$ as $T$ is not empty and necessarily $T\subseteq A$, because of the minimality of $T$. Then, consider any deal $(A,B,C)$ where $B$ is an arbitrary subset of $S\cup T\setminus A$ and $C$ is the rest of the cards. In that deal, if Alice announces that $\sum_d A=x$ then Cath knows that all the cards in $T$ belong to $A$ because all the cards that Cath does not hold are in $S\cup T$ and the only way of forming a set of $a$ cards with sum $x$ is by using all the cards in $T$ (this by minimality).

Therefore the announcement of $\sum_dA$ is not secure, since Cath would learn at least one of Alice's cards. If the corresponding condition for $b$ fails, then by the same argument Cath would learn one of Bob's cards. In either case we conclude that the announcement of $\sum_dA$ is not secure.
\end{proof}

To obtain our main result, Theorem \ref{corPrime}, we now combine this proposition with a combinatorial theorem, conjectured by Erd\"os and Heilbronn in \cite{eh} and proven by Dias da Silva and Hamidoune in \cite{dh}: 

\begin{proposition}[\cite{dh}]\label{one} Let $d$ be a prime. For a set $A\subseteq\mathbb Z/(d)$, denote by $S^n(A)$ the set of all sums $x_1+...+x_n$ of $n$ distinct elements of $A$. Then, \[|S^n(A)|\geq\min \{ d,n|A|-n^2+1 \} .\] \end{proposition} 

\begin{theorem}[Security] \label{corPrime} If $a+b+c$ is prime, then announcing $\sum_d A$ (or $\sum_d B$) is secure for all deals of size $(a,b,c)$ if and only if \[ \begin{array}{rcl}
    ab - 2a - b - c + 1 & \geq & 0 \ \text{and}\\
    ab - 2b - a - c + 1 & \geq & 0
  \end{array} \] 
\end{theorem}

\begin{proof}
By Proposition \ref{teoAB} we need that any set $X$ of $a+b-1$ cards has subsets of size $a$ and $b$ which sum up to all the possible values from $0$ to $d-1$. This condition holds for $a$ whenever
\[ \begin{array}{c}
    |S^a (X)| = d,
  \end{array} \] 
which in turn holds whenever
\[a|X|-a^2+1\geq d,\]
that is, when
\[a(a+b-1)-a^2+1\geq a+b+c.\]
Some simplification takes this inequality to the form displayed above.

Note that the bound is tight because the inequality in Proposition \ref{one} is also tight; we have that
\[|S^a(\{1,2,\dots,a+b-1\})|=a(a+b-1)-a^2+1,\]
so if the latter is less than $d$ we have that the announcement of $\sum_dA$ is not always secure (by the `only if' direction of Proposition \ref{teoAB}).

A similar analysis shows that the second inequality is equivalent to Cath always being ignorant of Bob's cards.
\end{proof}

\begin{corollary} \label{corC1} In deals of size $(a,b,1)$ with $a+b+1$ prime, announcing $\sum_dA$ is secure if and only if $a,b>2$.
\end{corollary}

\begin{proof}
When $c=1$, the requisites in Theorem \ref{corPrime} are satisfied iff $a,b>2$.
\end{proof}
This gives us a full characterization for when the $\mathsf{DModSum}$ protocol works when $c=1$ and $d$ is prime. In the following section we will show how one can deal with non-prime values of $d$.

\section{The least-prime modulo-sum announcement} \label{two}

The least-prime modulo-sum protocol is the $n$-$\mathsf{ModSum}$ protocol for $n$ the least prime greater than or equal to $d$. In this section we show that for $c=1$ and $a,b \geq 3$ except for $a=4$ and $c=3$, and $a=4$ and $c=3$, the least-prime modulo-sum protocol is secure. Consider a card deal with size parameters $(a,b,c)$. 

\begin{definition}[$\mathsf{LPModSum}$] 
The $\mathsf{LPModSum}$ (`least-prime modulo-sum') protocol is as follows: Alice and Bob announce the sum of their cards modulo $p$, where $p$ is the least prime greater than or equal to $a+b+c$.
\end{definition}









We are now ready to obtain our result for the $\mathsf{LPModSum}$ protocol for $(a,b,1)$. The idea is to find some appropriate $c \geq 1$ such that $a+b+c$ is prime and $(a,b,c)$ satisfies the conditions of the Security Theorem \ref{corPrime}, then apply the Reduction Proposition \ref{reduction}. First we need the following properties on the distribution of primes:

\begin{proposition}[Bertrand's postulate, \cite{bertrand}] For each $n > 3$, there is always a prime $p$ such that $n < p < 2n -2$.
\end{proposition}

\begin{proposition}[Nagura's theorem, \cite{nagura:1952}] For each $n > 24$, there is always a prime $p$ such that $n < p < \frac{6}{5} n$.
\end{proposition}

\begin{theorem}\label{T1} For all $a,b \geq 3$ except for $a=4, b=3$ and $a=3, b=4$, the $\mathsf{LPModSum}$ protocol is secure for size $(a,b,1)$.
\end{theorem}

\begin{proof} We may assume that $a \geq b$, as security for $(a,b,c)$ entails security for $(b,a,c)$.

\medskip

\noindent \underline{Case 1}: $b > 3$.
\smallskip

\noindent \underline{Case 1.1}: $a+b+1$ is prime.
\smallskip

\noindent Then, the $\mathsf{DModSum}$ and the $\mathsf{LPModSum}$ protocols coincide and so the security of the protocol follows from Corollary \ref{corC1}.

\smallskip
\noindent \underline{Case 1.2}: $a+b+1$ is not prime.

\smallskip
\noindent Let $p$ be the least prime number greater than $a+b+1$ and put $c= p- (a+b)$. By Bertrand's postulate for $n=a+b$ we get that there is a prime $q$ such that $a+b < q < 2(a+b)-2$ and hence $p \leq 2(a+b)-3$. Then, $a+b+c$ is prime and it is easy to check that $(a,b,c)$ satisfies the conditions of the Security Theorem \ref{corPrime} (observe that we only need to verify the first condition since we are assuming $a \geq b$).

Namely,
$$\begin{array}{lcl}
ab-2a-b-c+1 & = & ab -2a -b-(p-(a+b)) + 1 \\
  & = & ab -a -p + 1 \\
  & \geq & ab-a-(2(a+b)-3) + 1\\
  & = & ab-3a-2b+4.
  \end{array}$$


\begin{enumerate}

 \item[] \underline{Case A}: $b \geq 5$. \newline Then, we have (recall that $a \geq b$):
 $$ab-3a-2b+4 \geq ab-3a-2a+4 = ab-5b + 4 \geq 5a-5a+ 4 = 4 \geq 0.$$

\item[] \underline{Case B}: $b =4 $. \newline Then, we have (recall that $a \geq b = 4$):

$$ab-3a-2b+4 = 4a-3a-8+4 = a-4 \geq 0.$$

\end{enumerate}

So, it follows from the Security Theorem \ref{corPrime} that the modulo-sum protocol is secure for size $(a,b,p-(a+b))$ and then the least-prime modulo-sum protocol for size $(a,b,1)$ is also secure by the Reduction Proposition \ref{reduction}.
\medskip

\noindent \underline{Case 2}: $b = 3$.

\smallskip
\noindent \underline{Case 2.1}: $a+4$ is prime.

\smallskip

\noindent  Then, the $\mathsf{DModSum}$ and the $\mathsf{LPModSum}$ protocols coincide and so the security of the protocol follows from Corollary \ref{corC1}.

\smallskip

\noindent \underline{Case 2.2}: $a+4$ is not prime.

\smallskip

\noindent Let $p$ be the least prime number greater than $a+4$ and put $c = p-(a+3)$. Reasoning as in case 1.2, it follows from Bertrand's postulate that $p \leq 2a+3$. Nevertheless, now this upper bound for $p$ does not allow us to verify that the conditions of the Security Theorem \ref{corPrime} hold for size $(a,3,c)$. Namely,

$$\begin{array}{lcl}
ab-2a-b-c+1 = 3a - 2a -3 -(p-(a+3)) + 1  & = &  2a-p+1  \\ & \geq
&2a -(2a+3) + 1   =  -2. \end{array}$$

\noindent However, observe that in order to obtain a value greater than or equal to 0 in the calculation above it would suffice to improve our upper bound on the prime $p$ to $p \leq 2a+1$. We will be able to do this by using Nagura's theorem for $a > 21$ and by inspection for $4< a \leq 20$.

\begin{itemize}

\item[] \underline{Case A}: $a > 21$. \newline
It follows from Nagura's theorem for $n = a+3$ that $p \leq \frac{6}{5} (a+3)$. But then:
$$p \leq \frac{6}{5} (a+3) = \frac{6a+18}{5} \leq \frac{6a+a}{5} = \frac{7}{5}a \leq 2a + 1.$$
\item[] \underline{Case B}: $4< a \leq 20$. \newline

We can readily compute the actual values of the primes $p$ for all $4 < a \leq 20$ with $a+4$ not prime. Namely, \\

\begin{center}

\begin{tabular}{|c|c|c|c|c|c|c|c|c|c|c|c|} \hline
 $a$  &5 &6 &8 &10 &11 &12 &14 &16 &17 & 18 & 20 \\ \hline
 $p$  &11 &11 &13 &17 &17 &17 &19 &23 &23 &23  &29 \\ \hline
 $2a+1$  &11 &13 &17 &21 &23 &25 &29 &33 &35 &37  & 41 \\ \hline
\end{tabular}
\end{center}

\noindent \\ Note that $p \leq 2a+1$ holds for all such $a$'s.
\end{itemize}

\medskip
\noindent This completes the proof of the theorem.
\end{proof}

\section{Special cases and final results} \label{341}

We have not dealt with the special cases $(4,3,1)$ and $(3,4,1)$ yet. Let us analyze the case $(4,3,1)$. For these size parameters, the $\mathsf{LPModSum}$ protocol amounts to the announcement of $\sum_{11} A$, whereas the $\mathsf{DModSum}$ protocol amounts to the announcement of $\sum_8 A$. We have checked using Haskell (see the Appendix) that:

\begin{itemize}
\item the $\mathsf{LPModSum}$ protocol is not secure; for example, if Alice holds $0123$, Bob holds $456$ and Cath holds 7, then after the announcement of $\sum_{11} A = 6$, Cath learns that Alice holds card 2; but

\item the $\mathsf{DModSum}$ protocol is secure.

\end{itemize}

\noindent By symmetry, similar remarks apply for the $(3,4,1)$ case. We recall that the $(3,3,1)$ case is an instance of the $\mathsf{DModSum}$ protocol, namely $7$-$\mathsf{ModSum}$. We can now combine our previous results to obtain the following:

\begin{theorem} For $a,b \geq 3$, there is a secure and informative protocol for deals of size $(a,b,1)$.
\end{theorem}

\section{Conclusions and future work}

Given three card players Alice, Bob, and Cath who hold $a,b,c$ cards, respectively, and who only know their own cards, we have investigated general conditions for when Alice or Bob can safely announce the sum of the cards they hold modulo an appropriately chosen integer. We demonstrate that this holds whenever $a,b>2$ and $c=1$. Because Cath holds a single card, this also implies that Alice and Bob will learn the card deal from the other player's announcement. We want to continue our investigations for cases where $c > 1$, and compare this to incidental results from the literature. The authors believe that their work could lead to new methods of secure information exchange.

\bibliographystyle{plain}

\bibliography{biblio2011a}

\section*{Appendix: the Haskell script {\tt subsets.hs}}

\begin{figure}[h]
{\footnotesize
\begin{verbatim}
import Data.List
-- (subsets n xs) outputs the list of all the subsets of xs of n elements. 
subsets :: Int -> [Int] -> [[Int]]
subsets 0 _ = [[]]
subsets _ [] = []
subsets (n+1) (x:xs) = [x:ys | ys <- subsets n xs] ++ subsets (n+1) xs
-- (subsetSum m n xs) outputs the list of all sums (modulo m) of the subsets of xs of n elts.
subsetSum :: Int -> Int-> [Int] -> [Int]
subsetSum m n xs =  nub [mod (sum ys) m | ys <- subsets n xs]
-- (deals a b c) generates all the deals in an (a,b,c) card distribution. 
deals a b c = [[xs,ys,zs] | xs <- subsets a [0..(a+b+c-1)],
                            ys <- subsets b ([0..(a+b+c-1)] \\ xs),
                            zs <- [([0..(a+b+c-1)] \\ xs) \\ ys]]
-- (check m n as bs) checks whether each card of as can be interchanged in an n-tuple with an 
-- n-tuple of elements of bs with the same sum (modulo m).
check m n as bs = and [ or [elem (mod (x+y) m) ys | y <- subsetSum m (n-1) (as \\ [x])]| x <- as]
                   where ys = subsetSum m n bs
-- (secure a b c n) checks whether for each deal of an (a,b,c) card distribution, [as,bs,cs],  
-- each card of as can be interchanged in an n-tuple with an n-tuple of elements of bs with 
-- the same sum (modulo a+b+c).   
secure a b c n = and [check (a+b+c) n as bs | [as,bs,_] <- deals a b c]		   
-- (secure2 a b c) checks whether for each deal of an (a,b,c) card distribution, [as,bs,cs], 
-- there exists some n <= min(a,b) such that each card of as can be interchanged in an n-tuple
-- with an n-tuple of elements of bs with the same sum (modulo a+b+c)"
secure2 a b c = and [or [check (a+b+c) n as bs | n <- [2..min a b] ]| [as,bs,_] <- deals a b c]   
\end{verbatim}
}
\caption{The Haskell script {\tt subsets.hs}}
\label{fig.haskell}
\end{figure}

For input card deal of size $(3,4,1)$ the Haskell script {\tt subsets.hs} in Figure \ref{fig.haskell} shows that the $d$-$\mathsf{ModSum}$ protocol is secure. The implemented algorithm is the natural brute force one. In the general setting of an $(a,b,c)$ card distribution, we firstly define a function
\medskip

\verb"deals :: Int -> Int -> Int ->  [[Int]]"

\medskip
\noindent so that \verb"deals a b c" generates all the ${{a+b+c} \choose{a}} \cdot {{b+c} \choose {b}}$ possible card deals in an $(a,b,c)$ card distribution. Next, we define an auxiliary predicate
\medskip

\verb"check :: Int -> Int -> [Int] -> [Int] ->  Bool"

\medskip
\noindent so that \verb"check d k as bs" checks whether each card of Alice's hand $as$ can be interchanged in an $k$--tuple with a $k$--tuple of elements of Bob's hand $bs$ with the same sum (modulo $d$). Finally, combining \verb"deals" and \verb"check" we define the main \textit{generate and test} predicate
\medskip

\verb"secure :: Int -> Int -> Int -> Int ->  Bool"

\medskip
\noindent so that \verb"secure a b c k" checks whether for each card deal of an $(a,b,c)$ card distribution each card of Alice's hand can be interchanged in a $k$--tuple with a $k$--tuple of elements of Bob's hand with the same sum (modulo $a+b+c$). For {\em card safety} we also have to check \verb"secure b a c k", namely that each card of Bob's hand can be interchanged in a $k$--tuple with a $k$--tuple of elements of Alice's hand with the same sum. 

\end{document}